\documentclass[ejsv2,preprint]{imsart}

\RequirePackage[numbers]{natbib}
\usepackage{titlesec}
\usepackage{orcidlink}
\usepackage{comment}
\newtheorem{prop}{Proposition}[section]
\newtheorem*{prop*}{Proposition}
\newtheorem{corollary}{Corollary}[section]
\newcommand{\Expect}[1]{\mathbb{E}\left[ #1 \right]}

\newcommand{\diffd}{\textnormal{d}}

\newcommand{\pderiv}[2]{\frac{\partial #1}{\partial #2}}

\newcommand{\abs}[1]{\left| #1 \right|}

\newcommand{\R}{\mathbb{R}}
\newcommand{\intR}{\int_{\R}}

\newcommand{\gbp}{generalized Beta-prime }

\newcommand{\tfo}[1]{{}_2F_1\left(#1\right)}
\begin{document}
%\linenumbers
\begin{frontmatter}
\title{The Continuous Rank Probability Score of a Generalized Beta-Prime Distribution and Some Special Cases}
%\title{A sample article title with some additional note\thanksref{t1}}
\runtitle{CRPS of a Generalized Beta-Prime}
%\thankstext{T1}{A sample additional note to the title.}

\begin{aug}
%%%%%%%%%%%%%%%%%%%%%%%%%%%%%%%%%%%%%%%%%%%%%%%
%% Only one address is permitted per author. %%
%% Only division, organization and e-mail is %%
%% included in the address.                  %%
%% Additional information can be included in %%
\author[A]{\fnms{Matthew}~\snm{LeDuc}\ead[label=e1]{matthew.leduc@colorado.edu}},
%\author[B]{\fnms{Keanan}~\snm{Gleason}\ead[label=e2]{keanan.gleason@colorado.edu}\orcid{0000-0000-0000-0000}}
%\and
%\author[B]{\fnms{Third}~\snm{Author}\ead[label=e3]{third@somewhere.com}}
%%%%%%%%%%%%%%%%%%%%%%%%%%%%%%%%%%%%%%%%%%%%%%
%% Addresses                                %%
%%%%%%%%%%%%%%%%%%%%%%%%%%%%%%%%%%%%%%%%%%%%%%
\address[A]{Department of Applied Mathematics,
University of Colorado Boulder\printead[presep={,\ }]{e1}}

%\address[B]{Department of Economics,
%University of Colorado Boulder\printead[presep={,\ }]{e2,e3}}
\runauthor{M. LeDuc}
\end{aug}

\begin{abstract}
This working paper describes new results in derivations of the Continuous Ranked Probability Score of a generalized beta-
prime distribution and several special cases, such as the Dagum distribution 
and Singh-Maddala distribution. Comparison with Monte Carlo estimates is also presented.
\end{abstract}



\end{frontmatter}

\section{Introduction}

Due to the significant increase in computational power over recent decades, and the subsequent rise in the popularity of Bayesian methods, there have been many discussions about the proper way to evaluate the performance of probabilistic models.
One of the more commonly used proper scoring rules is the Continuous Rank Probability Score (CRPS,\cite{MR4832261,scoring}), which is a strictly proper scoring rule given by
\begin{equation}\label{eq:crps_integraldef}
    CRPS(F|y) = \int_{\mathbb{R}}\left(F(x)-H(x-y)\right)^2\text{d}x
\end{equation}
where $F$ is the predictive cumulative density function, $y$ the true value, and $H$ the Heaviside step function. One attractive feature of the rule is that the CRPS has the same units as the quantity of interest, allowing easy interpretation. Another is that it generalizes the absolute error: if the forecast is a deterministic value $\hat{x}$, then $F(x)=H(x-\hat{x})$ and then $CRPS(F|y) = \abs{\hat{x}-y}$. This allows the CRPS to be used to compare probabilistic models with deterministic ones \cite{scoring, Mardani2025}. As a testament to this, the CRPS has been widely used, especially in atmospheric science and weather prediction (e.g. \cite{ahmed2024evaluation,brocker2012evaluating,Mardani2025,Price2025,crps_taillardat,zamo2018estimation}).

While the CRPS is easy to interpret due to its relationship with the absolute error in estimation, the integral in Eq. \eqref{eq:crps_integraldef} can be potentially difficult to calculate, especially for poorly behaved distributions. To this end, there are several closed forms for the CRPS for parametric distributions. Tables of some of them are given in  \cite{crps_taillardat} and \cite{zamo2018estimation}, and the \textit{R} package \textit{scoringRules} \cite{scoringRules_package} has numerical implementations. 

One parametric distribution for which there is no such form is the generalized Beta-prime distribution, also known as the generalized Beta of the second kind \cite{betaprimemcdonald, generalizedbeta}. This distribution, whose probability density function is given by
\begin{equation}
    \label{eq:genbp_pdf}
    p_Z(z) = \frac{p}{qB(\alpha, \beta)} \frac{(z/q)^{\alpha p-1}}{(1+(z/q)^p)^{\alpha+\beta}}
\end{equation}
arises from the ratio of two generalized Gamma variables \cite{malik1967exact} or from suitable transformations of the Beta distribution \cite{betaprimeregression}. 
This distribution has applications in economics, where it and its special cases often serve as distributions of income (e.g. \cite{dagum2008new,mahler2025parsimonious,singh2008function,burrdist}, and remote sensing, where it can be used as a model for distributions of ratios of photon intensities \cite{hardnessratio,mine_permproc,Park_2006,genbp_sar}. It has also seen use in probabilistic numerics \cite{chaskalovic2022gbp_finiteelements}.
Special cases of the distribution include the log-logistic distribution ($\alpha=\beta=1$,\cite{akhtar2014log}), the Singh-Maddala distribution ($\alpha=1$,\cite{kumar2017singh,singh2008function}), the Dagum distribution ($\beta=1,$ \cite{dagum2008new,kleiber2008guide}), and the generalized F distribution ( $p=1$, \cite{pham1989generalized}, of which the $F$ distribution corresponds to $q=\beta/\alpha$ ). Among these distributions, only the log-logistic distribution has a known closed form CRPS \cite{scoringRules_package,crps_taillardat}. In this work, we will derive the CRPS of the generalized Beta-prime distribution, and through it derive a CRPS for the Singh-Madalla and Dagum distributions. We will also demonstrate that this reduces to the log-logistic CRPS found in \cite{crps_taillardat} under the appropriate restrictions. % as well as an expression for the case where $p=1/N$ for $N\in \mathbb{N}$, which contains the generalized F distribution when $p=1$. %These results will be compiled into a table at the end of the manuscript (Table \textcolor{red}{add the table}), and the appendices will contain ways to derive the Dagum and Singh-Madalla CRPS directly. 

\section{The CRPS of a Generalized Beta-Prime Distribution}

Rather than using the integral formulation of the CRPS given by Eq. \eqref{eq:crps_integraldef}, closed form expressions of the CRPS are typically derived by considering an alternate form of the scoring rule. The following proof will rely on being able to write the CRPS as \cite{scoring}

\begin{equation}\label{eq:crps_alternate}
        CRPS(F|y) = \Expect{\abs{X-y}} - \Expect{X\left(2F(X)-1\right)}
    \end{equation}
which converts the integral form of the rule, which relies on evaluating the CDF of the predictive distribution, into one involving expectations that are relatively easier to manage with sufficient manipulation. Throughout, we will use the fact that $_2F_1(a,b;c;t)= {}_2F_1(b,a;c;t)$. Similar properties hold for the upper and lower parameters of higher order generalized hypergeometric functions as well.

\begin{prop}\label{prop:gbp_crps_prop}
  Suppose that the predictive distribution is a \gbp distribution with finite mean $\mu$. Then given observation $y$ the CRPS of the predictive distribution is given by
    \begin{equation}\label{eq:gbp_crps}
    \begin{split}
        &CRPS(F|y) = \\
        &2\mu - y +\frac{2}{B(\alpha, \beta)}\left(yB(w;\alpha,\beta)+\frac{y}{\alpha+\beta-1} w^{\alpha-1}(1-w)^\beta\left(1-\tfo{1,\alpha+\beta-1;\alpha+p^{-1},w }\right)  \right)\\
        &-\frac{2q}{\alpha B(\alpha,\beta)^2}B(2\alpha+p^{-1},2\beta-p^{-1}) {}_3F_2(\alpha+\beta, 1, 2\alpha+p^{-1};\alpha+1,2\alpha+2\beta;1)
        \end{split}
    \end{equation}
    where $B(x;\alpha,\beta)$ is the incomplete Beta function and $w=\frac{y^p}{q^p+y^p}$.
\end{prop}
\begin{proof}
    Starting from Eq. \eqref{eq:crps_alternate} we can say that 
    \begin{equation}
        \Expect{\abs{X-y}} =\mu-y+2\int_0^yF(x)\diffd x
    \end{equation}
    by splitting the integral appropriately and using integration by parts. Then making the substitution $u=\frac{x^p}{q^p+x^p}$, we get that
    \begin{equation}
    \begin{split}
         \int_0^yF(x)\diffd x &= \frac{q}{pB(\alpha \beta)}\int_0^wu^{1/p-1}(1-u)^{-1/p-1}B(u;\alpha,\beta)\diffd u \\
         & = \frac{q}{pB(\alpha,\beta)}\int_0^w\int_0^uu^{1/p-1}(1-u)^{-1/p-1}t^{\alpha-1}(1-t)^{\beta-1}\diffd t\diffd u
    \end{split}
    \end{equation}
    We can apply Fubini's Theorem to swap the order of integration, then making the substitution $s = \frac{u}{1-u}$ yields the integral
    \begin{equation}
    \begin{split}
        &\frac{q}{pB(\alpha,\beta)}\int_0^w\int_{\frac{t}{1-t}}^{\frac{w}{1-w}}s^{1/p-1}t^{\alpha-1}(1-t)^{\beta-1}\diffd s \diffd t \\
        =&\frac{q}{B(\alpha,\beta)}\left(\frac{y}{q}B(w;\alpha,\beta)-\int_0^wt^{\alpha+1/p-1}(1-t)^{\beta-1/p-1}\diffd t\right)\\
        =&\frac{q}{B(\alpha,\beta)}\left(\frac{y}{q}B(w;\alpha,\beta)-w^{\alpha+1/p-1}(1-w)^{\beta-1/p}\frac{w}{\alpha+1/p}\tfo{\alpha+\beta, 1;\alpha+1/p+1;w}\right)
        \end{split}
    \end{equation}
    This form allows us to take advantage of Eq. 15.5.16\_5 of the NIST Digital Library of Mathematical Functions (\cite{NIST:DLMF}), which says that
    \begin{equation}
        \label{dlmf15516_5}
     \tfo{a,b;c;z}-\tfo{a-1,b;c;z}-\frac{bz}{c}\tfo{a,b+1;c+1;z}=0
    \end{equation}
    This, along with the fact that hypergeometric functions are symmetric in their upper arguments, means that
    \begin{equation}
        -\frac{w}{\alpha+1/p}\tfo{\alpha+\beta, 1;\alpha+1/p+1;w} = \frac{1}{\alpha+\beta-1}\left(1-\tfo{1,\alpha+\beta-1;\alpha+1/p;w}\right)
    \end{equation}
    and so
    \begin{equation}
        \label{eq:gbp_crps_firstpart}
        \begin{split}
            \Expect{\abs{X-y}} &=\mu-y  +\frac{2}{B(\alpha,\beta)}\times  \\
            &\left(yB(w;\alpha,\beta) +yw^{\alpha-1}(1-w)^{\beta}\left(1-\tfo{1,\alpha+\beta-1;\alpha+1/p;w}\right)\right)
        \end{split}
    \end{equation}
    To finish the derivation, we must calculate $\Expect{X(2F(X)-1)}$. Since $-\Expect{X}=-\mu$, we focus on calculating $\Expect{2XF(X)}$. This is given by the integral
    \begin{equation}
        2\Expect{XF(X)}=\frac{p}{B(\alpha,\beta)^2}\intR \frac{(x/q)^{\alpha p}}{(1+(x/q)^p)^{\alpha+\beta}}B\left(\frac{(x/q)^p}{1+(x/q)^p};\alpha,\beta\right)\diffd x
    \end{equation}
    Now letting $u=\frac{(x/q)^p}{1+(x/q)^p}$ we see that this is equivalent to the integral
    \begin{equation}
        \frac{q}{B(\alpha,\beta)^2}\int_0^1u^{\alpha+1/p-1}(1-u)^{\beta-1/p-1}B(u;\alpha,\beta)\diffd u
    \end{equation}
    which can be evaluated in several ways. One way is given by Remark 3 of \cite{connor2022integrals}, but we will use an equivalent method. Via Eq. (8.17.8) of \cite{NIST:DLMF}, the integral becomes
    \begin{equation}
        \frac{q}{\alpha B(\alpha,\beta)^2}\int_0^1u^{2\alpha+1/p-1}(1-u)^{2\beta-1/p-1}\tfo{\alpha+\beta,1;\alpha+1;u}\diffd u
    \end{equation}
    which can be evaluated using Euler's Integral Transformation as (\cite{NIST:DLMF} Eq. (16.5.2))
    \begin{equation}
        \frac{q}{\alpha B(\alpha,\beta)^2}B(2\alpha+1/p,2\beta-1/p) {}_3F_2( \alpha+\beta, 1, 2\alpha+1/p;\alpha+1,2(\alpha+\beta);1 )
    \end{equation}
    By Eqs. \eqref{eq:crps_alternate} and \eqref{eq:gbp_crps_firstpart}, then, we see that 
     \begin{equation}
    \begin{split}
        &CRPS(F|y) = \\
        &2\mu - y +\frac{2}{B(\alpha, \beta)}\left(yB(w;\alpha,\beta)+\frac{y}{\alpha+\beta-1} w^{\alpha-1}(1-w)^\beta\left(1-\tfo{1,\alpha+\beta-1;\alpha+p^{-1},w }\right)  \right)\\
        &-\frac{2q}{\alpha B(\alpha,\beta)^2}B(2\alpha+p^{-1},2\beta-p^{-1}) {}_3F_2(\alpha+\beta, 1, 2\alpha+p^{-1};\alpha+1,2\alpha+2\beta;1)
        \end{split}
    \end{equation}
    as desired.
\end{proof}

In the following section, we will derive several special cases of this expression. We will start with obtaining expressions for the CRPS of the Dagum, Singh-Madalla, and log-logistic distributions, the latter of which is known from \cite{crps_taillardat}. %Then we will include a derivation of the CRPS of the \gbp distribution when $p^{-1}\in \mathbb{N}$.

\section{The CRPS of Several Special Cases of the Generalized Beta-Prime Distribution}

In this section we present derivations of the CRPS of several special cases of the generalized Beta-prime distribution as special cases of Proposition \ref{prop:gbp_crps_prop}. These can also be derived from Eq. \eqref{eq:crps_alternate} directly. We will derive scores for the Dagum and Singh-Maddala distributions, as well as provide an alternate way to verify the expression of the CRPS of a log-logistic distribution given by \cite{crps_taillardat}.%, and present a derivation for a special case of Proposition \ref{prop:gbp_crps_prop} when $1/p\in \mathbb{N}$, which contains the generalized F-distribution as a speciasl case.

\subsection{The CRPS of the Singh-Maddala Distribution}

Our first special case involves the Singh-Maddala distribution. This distribution is commonly used in economics \cite{kumar2017singh,singh2008function} and corresponds to the special case $\alpha=1$.

%The derivation from Eq. \eqref{eq:crps_alternate} is more straightforward, and is presented in Appendix \textcolor{red}{Appendix}. Here, we derive it as a special case of Proposition \ref{prop:gbp_crps_prop}.

\begin{corollary}\label{corr:sm_crps}
The CRPS of a Singh-Maddala distribution is given by
    \begin{equation}\label{eq:crps_sm}
    \begin{split}
         CRPS(F,y) &=q\frac{\Gamma\left(2\beta-\frac{1}{p}\right)\Gamma\left(1+\frac{1}{p}\right)}{\Gamma(2\beta)}+y\left(1-\frac{2q^{\beta p}}{(q^p+y^p)^\beta} {}_2F_1\left(1,
         \beta,1;1+\frac{1}{p};\frac{y^p}{y^p+q^p}\right)\right)
    \end{split}
\end{equation}
assuming that the mean is finite.
\end{corollary}
\begin{proof}
    Starting from Eq. \eqref{eq:gbp_crps} we can say that
   \begin{equation}
   \begin{split}
       &CRPS(F|y) =\\
       &2\mu-y + 2y\left(1-(1-w)^{\beta}+(1-w)^{\beta}(1-\tfo{1,\beta;1+p^{-1};w} \right) \\
       &-2q\beta B(2+p^{-1},2\beta-p^{-1}) {}_3F_2\left(1+\beta,1,2+p^{-1};2,2\beta+2;1\right)
       \end{split}
   \end{equation}
    by setting $\alpha=1$. To begin, we first focus on the terms that depend on $y$. This can be simplified down to 
    \begin{equation}\label{eq:sm_y_terms}
      y\left( 1-2(1-w)^\beta \tfo{1,\beta;1+p^{-1};w} \right)=y\left(1-\frac{2q^{p\beta}}{(q^p+y^p)^\beta } \tfo{1,\beta;1+p^{-1};\frac{y^p}{q^p+y^p}}  \right)
    \end{equation}
    Now we turn our attention to the terms that do not depend on $y$. These are given by 
    \begin{equation}\label{eq:sm_no_y_terms}
        2\mu-2q\beta B(2+p^{-1},2\beta-p^{-1}) {}_3F_2\left(1+\beta,1,2+p^{-1};2,2\beta+2;1\right)
    \end{equation}
    Since the mean of the Singh-Maddala distribution is given by $q\beta B(1+1/p,\beta-1/p)$ we can write this as
    \begin{equation}
        2q\beta\left( B(1+1/p,\beta-1/p)-B(2+p^{-1},2\beta-p^{-1}) {}_3F_2\left(1+\beta,1,2+p^{-1};2,2\beta+2;1\right) \right)
    \end{equation}
    Now we apply Euler's integral identity: We can write the hypergeometric term as
    \begin{equation}
    \begin{split}
{}_3F_2\left(1+\beta,1,2+p^{-1};2,2\beta+2;1\right) &=\int_0^1\tfo{1+\beta, 2+p^{-1};2\beta+2;t}\diffd t \\
&=\int_0^1 \frac{2\beta+1}{\beta(1+p^{-1})}\pderiv{}{t}\tfo{\beta,1+p^{-1};2\beta+1;t}\diffd t\\
&=\frac{2\beta+1}{\beta(1+p^{-1})} \left(\frac{ \Gamma(2\beta+1)\Gamma(\beta-p^{-1})}{\Gamma(\beta+1)\Gamma(2\beta-p^{-1})} -1\right)
\end{split}
    \end{equation}
    Plugging this back into Eq. \eqref{eq:sm_no_y_terms} and writing the Beta functions in terms of the Gamma function we see that the terms that do not depend on $y$ can be written as
    \begin{equation}
    \begin{split}
    &2q\beta\left(\frac{\Gamma(1+p^{-1})\Gamma(\beta-p^{-1})}{\Gamma(\beta+1)}-\frac{\Gamma(1+p^{-1})\Gamma(2\beta-p^{-1})}{\Gamma(2\beta+1)}\left(\frac{\Gamma(2\beta+1)\Gamma(\beta-p^{-1})}{\Gamma(\beta+1)\Gamma(2\beta-p^{-1})}-1
    \right)\right)\\
=& q\frac{\Gamma(1+p^{-1})\Gamma(2\beta-p^{-1})}{\Gamma(2\beta)}
    \end{split}
    \end{equation}
    which is exactly the first term in Corollary \ref{corr:sm_crps}. Combining this with Eq. \eqref{eq:sm_y_terms} gives the desired result.
\end{proof}

\subsection{The CRPS of the Dagum Distribution}

Our second special case involves the CRPS of the Dagum distribution. This distribution is commonly used in economics for income distribution modeling \cite{dagum2008new,kleiber2008guide} and corresponds to the special case $\beta=1$.

\begin{corollary}\label{corr:dagum_crps}
    The CRPS of a Dagum distribution is given by
    \begin{equation}
        \begin{split}
        &CRPS(F|y) = \\
        &2\mu - q\frac{\Gamma(1-p^{-1})\Gamma(2\alpha+p^{-1})}{\Gamma(2\alpha)} -y +2y \left(\frac{y^p}{q^p+y^p}\right)^{\alpha-1}\left(1-\frac{q^p}{q^p+y^p}\tfo{1,\alpha;\alpha+p^{-1},w }\right)\\
        \end{split}
    \end{equation}
    assuming that the mean is finite.
\end{corollary}
\begin{proof}
    Starting from Eq. \eqref{eq:gbp_crps}, we set $\beta=1$. This yields
     \begin{equation}
    \begin{split}
        &CRPS(F|y) = \\
        &2\mu - y +2\left(y w^\alpha+y w^{\alpha-1}(1-w)\left(1-\tfo{1,\alpha;\alpha+p^{-1},w }\right)  \right)\\
        &-2q\alpha B(2\alpha+p^{-1},2-p^{-1}) {}_3F_2(\alpha+1, 1, 2\alpha+p^{-1};\alpha+1,2\alpha+2;1)
        \end{split}
    \end{equation}
    Since the hypergeometric function here has an upper and lower argument that are the same we can write it as
    \begin{equation}
    \begin{split}
        &{}_3F_2(\alpha+1, 1, 2\alpha+p^{-1};\alpha+1,2\alpha+2;1) = \\ 
&\tfo{2\alpha+p^{-1},1;2\alpha+2;1}=\frac{\Gamma(2\alpha+2)\Gamma(1-p^{-1})}{\Gamma(2\alpha+1)\Gamma(2-p^{-1})}
    \end{split}
    \end{equation}
    by (\cite{NIST:DLMF} Eq. (15.4.20)), which is valid since the assumption that the mean is finite ensures $p>1$. Inserting this back into the equation and using properties of the Beta and Gamma functions yields 
    \begin{equation}
        \begin{split}
        CRPS(F|y) &= 2\mu - q\frac{\Gamma(1-p^{-1})\Gamma(2\alpha+p^{-1})}{\Gamma(2\alpha)} -y \\&+2\left(y w^\alpha+y w^{\alpha-1}(1-w)\left(1-\tfo{1,\alpha;\alpha+p^{-1},w }\right)  \right)\\
        \end{split}
    \end{equation}
    Further algebra on the second portion of the equation yields
     \begin{equation}
        \begin{split}
        CRPS(F|y) &=2\mu - q\frac{\Gamma(1-p^{-1})\Gamma(2\alpha+p^{-1})}{\Gamma(2\alpha)} -y  \\
        &+2y \left(\frac{y^p}{q^p+y^p}\right)^{\alpha-1}\left(1-\frac{q^p}{q^p+y^p}\tfo{1,\alpha;\alpha+p^{-1},w }\right)\\
        \end{split}
    \end{equation}
    as desired.
\end{proof}
\begin{comment}
\subsection{The CRPS of a Generalized Beta-Prime when $1/p\in \mathbb{N}$}
\textcolor{red}{needs editing to bring in line with other notation}
\input{body.tex}
\end{comment}
\subsection{The CRPS of the Log-Logistic Distribution}

We can also derive the CRPS of the log-logistic CRPS as presented in \cite{crps_taillardat} from these results. They present the log-logistic CRPS in Appendix B.i as (adjusting for the notation used here)
\begin{equation}
    \label{eq:loglog_crps}
    CRPS(F|y) = \left(\frac{p-1}{p^2}\right)\frac{q\pi}{\sin(\pi/p)}+y\left(1-\frac{2q^p}{q^p+y^p} \tfo{1,1;1+p^{-1};\frac{y^p}{q^p+y^p}}\right)
\end{equation}
This can be obtained from Corollary \ref{corr:dagum_crps} by setting $\alpha=1$ or Corollary \ref{corr:sm_crps} by setting $\beta=1$. We will show this using the latter. Setting $\beta=1$ yields
\begin{equation}
    CRPS(F|y) = q\Gamma(2-p^{-1})\Gamma(1+p^{-1}) + y\left(1-\frac{2q^p}{q^p+y^p} \tfo{1,1;1+p^{-1};\frac{y^p}{q^p+y^p}}\right)
\end{equation}
Using properties of the Gamma function this can be written as
\begin{equation}
\begin{split}
    CRPS(F|y) =& q\left(\frac{p-1}{p}\right)\Gamma(1-p^{-1})\Gamma(1+p^{-1}) + y\left(1-\frac{2q^p}{q^p+y^p} \tfo{1,1;1+p^{-1};\frac{y^p}{q^p+y^p}}\right)\\
    =&\left(\frac{p-1}{p^2}\right)\frac{q\pi}{\sin(\pi/p)} + y\left(1-\frac{2q^p}{q^p+y^p} \tfo{1,1;1+p^{-1};\frac{y^p}{q^p+y^p}}\right)
    \end{split}
\end{equation}
assuming that $p>1$, which is necessary for the mean to exist. The last line is derived using Euler's reflection formula (\cite{NIST:DLMF} Eq. (5.5.3)). 
\section{Numerical Validation}

As a quick numerical validation of the method, we compared the results of the exact expression for the \gbp CRPS to Monte-Carlo simulation. By leveraging the relationship 
\begin{equation}
    X\sim B(\alpha,\beta) \implies q\left(\frac{X}{1-X}\right)^{1/p}\sim BP(\alpha, \beta, p, q)
\end{equation}
we were able to generate estimates of the CRPS via numerical integration of Eq. \eqref{eq:crps_alternate}. The results for $N=4\times 10^7$ draws are compared with the analytic expressions in Table \ref{tab:numerical}. The $N^{-1/2}$ convergence of Monte-Carlo methods suggests errors on the order of $10^{-4}$, which are observed.
\begin{table}[ht]
\centering

\begin{tabular}{cccccccr}
\hline
$\alpha$ & $\beta$ & $p$ & $q$ & $y$ & Analytical & Monte Carlo & Rel.\ Error \\
\hline
1.00 & 2.00 & 1.50 & 1.00 & 1.00 & 0.253261 & 0.253137 & $4.89 \times 10^{-4}$ \\
1.00 & 2.00 & 1.50 & 1.00 & 0.50 & 0.130956 & 0.131128 & $1.31 \times 10^{-3}$ \\
1.00 & 2.00 & 1.50 & 1.00 & 2.00 & 0.982212 & 0.982037 & $1.78 \times 10^{-4}$ \\
2.00 & 3.00 & 2.00 & 1.00 & 1.00 & 0.133398 & 0.133418 & $1.50 \times 10^{-4}$ \\
1.00 & 3.00 & 2.00 & 2.00 & 1.00 & 0.157655 & 0.157644 & $6.86 \times 10^{-5}$ \\
0.50 & 2.00 & 2.00 & 1.00 & 1.00 & 0.385010 & 0.384993 & $4.36 \times 10^{-5}$ \\
1.00 & 2.00 & 3.00 & 1.00 & 1.00 & 0.149604 & 0.149656 & $3.47 \times 10^{-4}$ \\
1.00 & 2.00 & 2.00 & 1.00 & 1.00 & 0.205476 & 0.205501 & $1.24 \times 10^{-4}$ \\
1.00 & 3.00 & 1.50 & 1.00 & 1.00 & 0.358729 & 0.358642 & $2.43 \times 10^{-4}$ \\
1.00 & 2.00 & 3.14 & 1.00 & 1.00 & 0.144072 & 0.144098 & $1.81 \times 10^{-4}$ \\
2.00 & 1.00 & 2.00 & 1.00 & 1.00 & 0.420078 & 0.420864 & $1.87 \times 10^{-3}$ \\
3.00 & 1.00 & 1.50 & 1.00 & 1.00 & 1.131873 & 1.132837 & $8.51 \times 10^{-4}$ \\
3.00 & 1.00 & 3.14 & 1.00 & 1.00 & 0.363000 & 0.362612 & $1.07 \times 10^{-3}$ \\
2.00 & 2.00 & 1.00 & 1.00 & 2.00 & 0.577778 & 0.577752 & $4.48 \times 10^{-5}$ \\
2.00 & 1.50 & 1.00 & 3.00 & 1.00 & 2.646148 & 2.644338 & $6.84 \times 10^{-4}$ \\
\hline
\end{tabular}
\caption{Numerical verification of the GBP CRPS formula against Monte Carlo estimation ($N = 4\times10^7$ samples).}
\label{tab:numerical}
\end{table}

%\section{Application to Space Weather Data}

%\cite{mine_permproc,cantrallmatsuo}

\section{Conclusion}

This short note, which will be expanded upon in the future, presents a derivation of the CRPS of a generalized Beta-prime distribution, along with the CRPS of the Dagum and Singh-Maddala distributions. This work shows that the derivbed expressions reduce to the CRPS of a log-logistic distribution under appropriate restrictions, and demonstrated a quick numerical verification of the results. Future work, which involves derivation of other special cases and applications to remote sensing and economics, is forthcoming.

\newpage
\bibliography{refs.bib, refs_crps.bib}
\end{document}